\newif\ifconf
\conffalse 

\ifconf
\documentclass[envcountsect,runningheads,a4paper]{llncs}
\else
\documentclass[11pt,envcountsect,runningheads]{llncs}
\fi

\usepackage
{
        amssymb,
        amsmath,
}

\usepackage{units}
\usepackage{graphicx}

\ifconf
\newcommand{\texorpdfstring}[2]{#1}
\else
\usepackage{fullpage}
\usepackage{etoolbox}
\newcounter{Hlemma}

\newcounter{Hremark}

\makeatletter
\g@addto@macro\lemma{\stepcounter{Hlemma}}
\g@addto@macro\remark{\stepcounter{Hremark}}
\let\llncs@addcontentsline\addcontentsline
\patchcmd{\maketitle}{\addcontentsline}{\llncs@addcontentsline}{}{}
\patchcmd{\maketitle}{\addcontentsline}{\llncs@addcontentsline}{}{}
\patchcmd{\maketitle}{\addcontentsline}{\llncs@addcontentsline}{}{}
\makeatother
\usepackage[unicode,colorlinks=true,citecolor=black,filecolor=black,linkcolor=black,urlcolor=black,pdfpagelabels,plainpages=false]{hyperref}
\fi

\def\compactify{\itemsep=0pt \topsep=0pt \partopsep=0pt \parsep=0pt}

\DeclareMathOperator{\argmax}{arg\,max}

\DeclareMathOperator {\Var}  {Var}

\DeclareMathOperator {\supp} {supp}

\newcommand {\set}   [1] {\left\{ #1 \right\}}
\newcommand {\brc}   [1] {\left(#1\right)}

\newcommand {\Exp}       {\mathbb{E}}
\newcommand {\Prob}  [1] {\Pr \brc{#1 }}

\newcommand {\E}     [1] {\Exp\left[#1\right]}
\newcommand {\EE}    [2] {\Exp_{#1}\left[#2\right]}

\newcommand {\calC} {{\cal C}}
\newcommand {\calI} {{\cal I}}
\newcommand {\CSPd} {MAX CSP${}_d$}
\newcommand {\kCSP} {MAX $k$-CSP}
\newcommand {\kCSPd} {MAX $k$-CSP${}_d$}

\title{Approximation Algorithm for Non-Boolean MAX $k$-CSP}
\author{Konstantin Makarychev\inst{1} \and Yury Makarychev\inst{2}\thanks{Yury Makarychev is supported in part by 
the NSF Career Award CCF-1150062.}}
\institute{Microsoft Research
\and Toyota Technological Institute at Chicago}
\date{}
\begin{document}
\maketitle

\begin{abstract}
In this paper, we present a randomized polynomial-time approximation algorithm for \kCSPd.
In \kCSPd, we are given a set of predicates of arity $k$ over an alphabet of size $d$.
Our goal is to find an assignment that maximizes the number of satisfied constraints.

Our algorithm has approximation factor $\Omega(kd/d^k)$ (when $k \geq \Omega(\log d)$).
This bound is asymptotically optimal assuming the Unique Games Conjecture.
The best previously known algorithm has approximation factor $\Omega({k\log d}/{d^k})$.

We also give an approximation algorithm for the boolean MAX $k$-CSP${}_2$
problem with a slightly improved approximation guarantee.
\end{abstract}

\section{Introduction}
We design an approximation algorithm for the \kCSPd, the maximum constraint satisfaction problem with $k$-ary
predicates and domain size $d$. In this problem, we are given a set of variables $\{x_u\}_{u\in X}$ and
a set of predicates $\cal P$. Each variable $x_u$ takes values in $[d]=\set{1,\dots,d}$. Each predicate
$P\in {\cal P}$ depends on at most $k$ variables. Our goal is to assign values to variables so as to maximize
the number of satisfied constraints.

There has been a lot of interest in finding the approximability of \kCSPd\ in the complexity community 
motivated by the connection of \kCSPd\ to $k$-bit PCPs. 
Let us briefly overview known results.
Samorodnitsky and Trevisan~\cite{ST00} showed that the boolean MAX $k$-CSP${}_2$
problem cannot be approximated within a factor of  $\Omega({2^{2\sqrt{k}}}/{2^k})$
if $P\neq NP$. Later Engebretsen and Holmerin~\cite{EH05} improved this bound to $\Omega({2^{\sqrt{2k}}}/{2^k})$.
For non-boolean \kCSPd, Engebretsen~\cite{Eng04} proved a hardness result of $2^{O(\sqrt{d})}/d^k$.
Much stronger inapproximability results were obtained assuming the Unique Games Conjecture (UGC).
Samorodnitsky and Trevisan~\cite{ST06} proved the hardness of $O(k/2^k)$ for
the boolean \kCSP${}_2$. Austrin and Mossel~\cite{AM} and, independently, Guruswami and Raghavendra~\cite{GR08}
proved the hardness of $O(kd^2/d^k)$ for non-boolean \kCSPd. Moreover, Austrin and Mossel~\cite{AM} proved 
the hardness of $O(kd/d^k)$ for every $d$ and infinitely many $k$; specifically,
their result holds for $d$ and $k$ such that $k=(d^t -1)/(d-1)$ for some $t\in{\mathbb{N}}$.
Very recently, H{\aa}stad strengthened the result of Austrin and Mossel and showed
the hardness of $O(kd/d^k)$ for every $d$ and $k \geq d$ [private communication].

On the positive side, approximation algorithms for the problem have been developed in a series of papers 
by Trevisan~\cite{Trev98},
Hast~\cite{Hast05}, Charikar, Makarychev and Makarychev~\cite{CMM06}, and 
Guruswami and Raghavendra~\cite{GR08}.
The best currently known algorithm for $k$-CSP$_d$ by  Charikar et al~\cite{CMM06} has
approximation factor of $\Omega({k\log d}/{d^k})$.
Note that a trivial algorithm for \kCSPd\ that just picks a random assignment satisfies each constraint with probability at least $1/d^k$,
and therefore its approximation ratio is $1/d^k$.

The problem is essentially settled in the boolean case. We know that
the optimal approximation factor is $\Theta(k/2^k)$ assuming UGC.
However, best known lower and upper bounds for the non-boolean case
do not match. In this paper, we present an approximation algorithm
for non-boolean \kCSPd\ with approximation factor $\Omega(kd/d^k)$
(for  $k \geq \Omega(\log d)$). This algorithm is asymptotically optimal
assuming UGC --- 
it is within a constant factor of the upper bounds 
of Austrin and Mossel and of H{\aa}stad 
(for $k$ of the form $(d^t -1)/(d-1)$ and for $k \geq d$, respectively).
Our result improves the best previously known
approximation factor of $\Omega(k\log d/d^k)$.


\textbf{Related work} Raghavendra studied a more general MAX CSP(${\cal P}$) problem~\cite{Rag08}. He showed that the optimal approximation factor equals the integrality gap 
of the standard SDP relaxation for the problem (assuming UGC). His result applies in particular to  \kCSPd.
However, the SDP integrality gap of \kCSPd\ is not known.

\textbf{Overview} We use semidefinite programming (SDP) to solve the problem. In our SDP relaxation,
we have an ``indicator vector'' $u_i$ for every variable $x_u$ and value $i$; we also have a ``indicator 
vector'' $z_C$ for every constraint $C$. In the intended solution, $u_i$ is equal to a fixed unit vector $\mathbf{e}$ if
$x_u=i$, and $u_i =0$ if $x_u \neq i$; similarly, $z_C = \mathbf{e}$ if $C$ is satisfied, and $z_C = 0$, otherwise.

It is interesting that the best previously known algorithm for the problem~\cite{CMM06}
did not use this SDP relaxation; rather it reduced the problem to a binary $k$-CSP problem, which  it solved 
in turn using semidefinite programming. The only previously known algorithm~\cite{GR08} that directly rounded
an SDP solution for \kCSPd\ had approximation factor $\Omega\left(\frac{k/d^7}{d^k}\right)$.

One of the challenges of rounding the SDP solution is that vectors $u_i$
might  have different lengths. Consequently, we cannot just use a rounding scheme that projects
vectors on a random direction and then chooses vectors that have largest projections, since this scheme
will choose longer vectors with disproportionately large probabilities. 
To deal with this problem, we first develop a rounding scheme that rounds
\textit{uniform} SDP solutions, solutions in which all vectors are ``short''. Then we construct a randomized
reduction that converts any instance to an instance with a uniform SDP solution.

Our algorithm for the uniform case is very simple. First, we choose a random Gaussian vector $g$.
Then for every $u$, we find $u_i$ that has the largest projection on $g$ (in absolute value), and
let $x_u = i$. However, the analysis of this algorithm is quite different from analyses of 
similar algorithms for other problems: when we estimate the probability that a constraint $C$ is satisfied,
we have to analyze the correlation of all vectors $u_i$ with vector $z_C$ (where $\{u_i\}$ are SDP vectors for variables $x_u$
that appear in $C$, $z_C$ is the SDP vector for $C$), whereas the standard approach would be
to look only at pairwise correlations of vectors $\{u_i\}$; this approach does not work in our case, however,
since vectors corresponding to an assignment that satisfies $C$ may have very small pairwise correlations,
but vectors corresponding to assignments that do not satisfy $C$ may have much larger pairwise correlations.

\begin{remark}
We study the problem only in the regime when $k\geq \Omega(\log d)$. In Theorem~\ref{thm:main}, 
we prove that when
$k = O(\log d)$ our algorithm has approximation factor
$e^{\Omega(k)}/d^k$. However, in this regime, a better approximation factor of $\Omega(d/d^k)$ 
can be obtained by a simple greedy approach.
\end{remark}

\textbf{Other Results} We also apply our SDP rounding technique to the Boolean Maximum CSP Problem.
We give an algorithm that has approximation guarantee ${}\approx 0.62\, k /2^k$ for sufficiently large $k$.
That slightly improves the best previously known guarantee of ${}\approx 0.44\, k/2^k$~\cite{CMM06}.
\ifconf
\else
We present this result in Appendix~\ref{sec:boolean-CSP}.
\fi

\section{Preliminaries} \label{sec:prelim}
We apply the approximation
preserving reduction of Trevisan~\cite{Trev98} to transform a general instance of \kCSPd\ to an instance where each 
predicate is a conjunction of terms of the form $x_u = i$. The 
reduction replaces a predicate $P$, which depends on variables $x_{v_1}$, \dots, $x_{v_k}$, with a set 
of clauses 
$$\set{(x_{v_1} = i_1) \wedge \dots \wedge (x_{v_k} = i_k): P(i_1, \dots, i_k) \text{ is true}}.$$
Then it is sufficient to solve the obtained instance.
We refer the reader to~\cite{Trev98} for details. We assume below that each predicate is a clause of
the form $(x_{v_1} = i_1) \wedge \dots \wedge (x_{v_k} = i_k)$.

\begin{definition}[Constraint satisfaction problem] An instance $\calI$ of \CSPd\ consists of 
\begin{itemize}
\item a set of ``indices'' $X$,
\item a set of variables $\{x_u\}_{u \in X}$ (there is one variable $x_u$ for every index $u\in X$),
\item a set of clauses $\calC$.
\end{itemize}
Each variable $x_u$ takes values in the domain $[d] = \{1,\dots, d\}$. Each clause $C\in \calC$ 
is a set of pairs $(u, i)$ where $u\in X$ and $i \in [d]$. 
An assignment $x_u = x_u^*$ satisfies a clause $C$ if for every $(u,i) \in C$, we have $x_u^* = i$.
We assume that no clause $C$ in $\calC$ contains pairs $(u,i)$ and $(u, j)$ with $i\neq j$ (no assignment satisfies such clause).
The length of a clause $C$ is $|C|$. The support of $C$ is $\supp(C) = \set{u: (u,i)\in C}$.

The value of an assignment $x_u^*$ is the number of constraints in $\calC$ satisfied by $x_u^*$. 
Our goal is to find an assignment of maximum value. We denote the value of an optimal assignment by $OPT = OPT(\calI)$.

In the \kCSPd\ problem, we additionally require that all clauses in $\calC$ have length at most~$k$.
\end{definition}

We consider the following semidefinite programming (SDP) relaxation for \CSPd. For every index $u\in X$ and $i\in [d]$,
we have a vector variable $u_i$; for every clause $C$, we have a vector variable~$z_C$.
\begin{align*}
\text{maximize:\ } &\sum_{C\in\calC} \|z_C\|^2\\
\text{subject to} & {}\\
&\sum_{i=1}^d  \|u_i\|^2 \leq 1 &&\quad\text{ for every } u\in X\\
&\langle u_i, u_j \rangle = 0 &&\quad\text{ for every } u\in X, i,j\in [d] \  (i\neq j)\\
&\langle u_i, z_C\rangle = \|z_C\|^2 &&\quad\text{ for every } C\in\calC,\  (u, i)\in C\\
&\langle u_j, z_C\rangle = 0 &&\quad\text{ for every } C\in\calC,\  (u, i)\in C \text{ and } j \neq i\\
\end{align*}
Denote the optimal SDP value by $SDP = SDP(\calI)$.
Consider the optimal solution $x_u^*$ to an instance $\calI$ and the corresponding SDP solution defined as follows, 
\begin{equation*}
u_i = 
\begin{cases}
\mathbf{e}, &\text{ if } x_u^* = i;\\
0, &\text{ otherwise};
\end{cases} \qquad\qquad
z_C = 
\begin{cases}
\mathbf{e}, &\text{ if } C\text{ is satisfied};\\
0, &\text{ otherwise};
\end{cases}
\end{equation*}
where $\mathbf{e}$ is a fixed unit vector. It is easy to see that this is a feasible SDP solution and its value equals $OPT(\calI)$. Therefore,
$SDP(\calI) \geq OPT(\calI)$.

\begin{definition}
We say that an SDP solution is uniform if $\|u_i\|^2\leq 1/d$ for every $u\in X$ and $i\in [d]$.
\end{definition}

\begin{definition}
Let $\xi$ be a standard Gaussian variable with mean $0$ and variance $1$. We denote 
\begin{align*}
\Phi(t) &= \Prob{|\xi| \leq t} = \frac{1}{\sqrt{2\pi}}\int_{-t}^te^{-x^2/2}dx, \text{ and}\\
\bar{\Phi}(t) &= 1 - \Phi(t) = \Prob{|\xi| > t}.
\end{align*}
\end{definition}
We will use the following lemma, which we prove in Appendix.
\begin{lemma}\label{lem:gaussian-scaling}
For every $t > 0$ and $\beta \in (0,1]$ , we have
$$\bar{\Phi}(\beta t) \leq \bar{\Phi}(t)^{\beta^2}.$$
\end{lemma}

We will also use the following result of \v{S}id\'{a}k~\cite{Sidak}:
\begin{theorem}[\v{S}id\'{a}k~\cite{Sidak}]
\label{thm:sidak}
Let $\xi_1,\ldots, \xi_r$ be Gaussian random variables with mean zero and an arbitrary covariance matrix. Then for any positive
$t_1, \ldots, t_r$,
$$\Prob{|\xi_1| \leq t_1,|\xi_2| \leq t_2, \ldots, |\xi_r| \leq t_r}
\geq \prod_{i=1}^r\Prob{|\xi_i| \leq t_i}.$$
\end{theorem}

\section{Rounding Uniform SDP Solutions}
In this section, we present a rounding scheme for uniform SDP solutions.
\begin{lemma}\label{lem:uniform-rounding}
There is a randomized polynomial-time algorithm that given an instance $\calI$ of the \CSPd\ problem (with $d\geq 57$) and a uniform SDP solution,
outputs an assignment $x_u$ such that for every clause $C \in \calC$:
$$\Prob{C \text{ is satisfied by } x_u} \geq \frac{\min(\|z_C\|^2 |C| d/8, e^{|C|})}{2d^{|C|}}.$$
\end{lemma}
\begin{proof}
We use the following rounding algorithm:

\rule{0pt}{12pt}
\hrule height 0.8pt
\rule{0pt}{1pt}
\hrule height 0.4pt
\rule{0pt}{6pt}

\noindent \textbf{Rounding Scheme for Uniform SDP solutions}
\medskip

\noindent \textbf{Input:} an instance of the \CSPd\ problem and a uniform SDP solution.

\noindent \textbf{Output:} an assignment $\{x_u\}$.
\begin{itemize}\compactify
\item Choose a random Gaussian vector $g$ so that every component of $g$ is distributed as a Gaussian variable with mean 0 and variance 1,
and all components are independent.
\item For every $u\in V$, let $x_u'= \argmax_i |\langle u_i, g\rangle|$. 
\item For every $u\in V$, choose $x_u''$ uniformly at random from $[d]$ (independently for different $u$).
\item With probability $1/2$ return assignment $\set{x_u'}$; with probability $1/2$ return assignment $\set{x_u''}$.
\end{itemize}
\rule{0pt}{1pt}
\hrule height 0.4pt
\rule{0pt}{1pt}
\hrule height 0.8pt
\rule{0pt}{12pt}

For every clause $C$, let us estimate the probabilities that assignments $x'_u$ and $x''_u$ satisfy $C$.
It is clear that $x_u''$ satisfies $C$ with probability $d^{-|C|}$. We prove now that 
$x_u'$ satisfies $C$ with probability at least $d^{-3|C|/4}$ if $\|z\|_C^2 \geq 8/(|C|d)$.
\begin{claim}
Suppose $C \in \calC$ is a clause such that $\|z\|_C^2 \geq 8/(|C|d)$ and $d\geq 57$. Then
the probability that the assignment $x_u'$ satisfies $C$ is at least $d^{-3|C|/4}$.
\end{claim}
\begin{proof}
Denote $s = |C|$.
We assume without loss of generality that for every $u\in \supp(C)$, $(u, 1) \in C$.
Note that for $(u,i)\in C$, 
we have $\|z_C\|^2 = \langle z_C, u_i\rangle \leq \|z_C\| \cdot \|u_i\| \leq \|z_C\|/\sqrt{d}$
(here we use that the SDP solution is uniform and therefore $\|u_i\|^2 \leq 1/d$).
Thus $\|z_C\|^2 \leq 1/d$. In particular, $s=|C| \geq 8$ since  $\|z\|_C^2 \geq 8/(|C|d)$.

For every $u \in \supp(C)$, let $u_1^{\perp} = u_1 - z_C$. 
Let $\gamma_{u,1} = \langle g, u_1^\perp\rangle$ and
$\gamma_{u,i} = \langle g, u_i\rangle$ for $i \geq 2$. Let $\gamma_C = \langle g, z_C\rangle$.
All variables $\gamma_{u,i}, \gamma_C$ are Gaussian variables.
Using that for every two vectors $v$ and $w$,
$\E{\langle g, v\rangle \cdot \langle g, w\rangle} = \langle v, w\rangle$, we get
%
\begin{align*}
\E{\gamma_C \cdot \gamma_{u,1}} &= \langle z_C, u_1 - z_C\rangle = \langle z_C, u_1\rangle - \|z_C\|^2 = 0;\\
\E{\gamma_C \cdot \gamma_{u,i}} &= \langle z_C, u_i\rangle = 0 \quad \text{for } i\geq 2.
\end{align*}
Therefore, all variables $\gamma_{u,i}$ are independent from $\gamma_C$. (However, for $u', u'' \in \supp(C)$
variables $\gamma_{u',i}$ and $\gamma_{u'',j}$ are not necessarily independent.)
Let $M = {\bar \Phi}^{-1}(1/d^{s/2}) / \sqrt{sd/8}$.
We write the probability that $x_u'$ satisfies $C$,
\ifconf
\else
\begin{align*}
  \Prob{x_u' \text{ satisfies } C} &= \Pr\bigl(\argmax_i |\langle g, u_i\rangle| = 1 \text{ for every } u \in \supp(C)\bigr)\\
&= \Prob{|\langle g, u_1\rangle| > |\langle g, u_i\rangle| \text{ for every } u \in \supp(C),i\in\set{2,\dots,d}}\\
&= \Prob{|\gamma_{u,1} + \gamma_C| > |\gamma_{u,i}| \text{ for every } u \in \supp(C), i\in\set{2,\dots,d}}\\
&\geq \Pr(|\gamma_{u,1}| \leq M/2, \text{ and } |\gamma_{u,i}| \leq M/2 \\
&\phantom{{}={}{}={}} \text{ for every } u \in \supp(C), i\in\set{2,\dots,d} \; \bigl| \; |\gamma_C| > M\bigr.)
\cdot \Prob{|\gamma_C| > M}.
\end{align*}
\fi

Since all variables $\gamma_{u,i}$ are independent from $\gamma_C$, 
\ifconf
\else
$$\Prob{x_u' \text{ satisfies } C} \geq 
\Prob{|\gamma_{u,i}| \leq M/2 \text{ for every } u \in \supp(C), i\in\set{1,\dots,d}} \cdot \Prob{|\gamma_C| > M}.$$
\fi
By \v{S}id\'{a}k's Theorem (Theorem~\ref{thm:sidak}), we have
\begin{equation}\label{eq:main}
\Prob{x_u' \text{ satisfies } C} \geq 
\Bigl(\prod_{u\in\supp(C)}\prod_{i=1}^d \Prob{|\gamma_{u,i}| \leq M/2}\Bigr) \cdot \Prob{|\gamma_C| > M}.
\end{equation}
We compute the variance of vectors $\gamma_{u,i}$. 
We use that $\Var[\langle g, v\rangle] = \|v\|^2$ for every vector $v$ and that the SDP solution is uniform.
\ifconf
\else
\begin{align*}
\Var[\gamma_{u,1}] &= \|u_1^{\perp}\|^2 = \|u_1 - z_C\|^2 = \|u_1\|^2 - 2\langle u_1, z_C\rangle + \|z_C\|^2 = \|u_1\|^2 - \|z_C\|^2 \leq \|u_1\|^2 \leq 1/d;\\
\Var[\gamma_{u,i}] &= \|u_i\|^2 \leq 1/d \quad\text{ for } i\geq 2.
\end{align*}
\fi
Hence since $\Phi(t)$ is an increasing function and ${\bar\Phi}(\beta t) \leq {\bar\Phi}(t)^{\beta^2}$ (by Lemma~\ref{lem:gaussian-scaling}),
we have
\begin{align*}
\Prob{|\gamma_{u,i}| \leq M/2} &= \Phi(M/(2\sqrt{\Var[\gamma_{u,i}]})) \geq \Phi(\sqrt{d}M/2) = 1 - {\bar\Phi}(\sqrt{d}\, M/2) \\
&\geq 1 - {\bar\Phi}(\sqrt{sd/8}\, M)^{2/s} = 1 - (d^{-s/2})^{2/s} = 1 - d^{-1}
\end{align*}
(recall that we defined $M$ so that ${\bar\Phi}(\sqrt{sd/8}\, M) = d^{-s/2}$).
Similarly, $\Var[\gamma_C] = \|z_C\|^2 \geq 8/(sd)$ (by the condition of the lemma). We get 
(using the fact that ${\bar\Phi}(t)$ is a decreasing function),
$$\Prob{|\gamma_C| > M} = {\bar\Phi}(M/\sqrt{\Var[\gamma_C]}) \geq {\bar\Phi}(M\sqrt{sd/8}) = d^{-s/2}.$$

Plugging in bounds for $\Prob{|\gamma_{u,i}| \leq M/2}$ and $\Prob{|\gamma_C| > M}$ into (\ref{eq:main}), we obtain
$$
\Prob{x_u' \text{ satisfies } C} \geq 
(1 - d^{-1})^{ds} d^{-s/2} \geq d^{-3s/4}.
$$
Here, we used that $(1 - d^{-1})^d \geq d^{-1/4}$ for $d\geq 57$ (the inequality $(1 - d^{-1})^d \geq d^{-1/4}$ holds for $d\geq57$
since it holds for $d=57$ and the left hand side, $(1 - d^{-1})^d$, is an increasing function, the right hand side, $d^{-1/4}$, is a decreasing function).
\qed
\end{proof}

We conclude that if $\|z_C\|^2 \leq 8/(|C|d)$ then the algorithm chooses assignment $x''_u$ with probability $1/2$ and this 
assignment satisfies $C$ with probability at least $1/d^{|C|} \geq \|z_C\|^2 \, |C|\, d / (8 \, d^{|C|})$. So
$C$ is satisfied with probability at least, $1/d^{|C|} \geq \|z_C\|^2 \, |C|\, d / (16 \, d^{|C|})$; if $\|z_C\|^2 \geq 8/(|C|d)$ then
the algorithm chooses assignment $x'$ with probability $1/2$ and this assignment satisfies $C$ with probability at least 
$d^{-3|C|/4} \geq e^{|C|}/d^{|C|}$ (since $e \leq 57^{1/4} \leq d^{1/4}$).
In either case, 
$$\Prob{C \text{ is satisfied}} \geq \frac{\min(\|z_C\|^2 |C| d/8, e^{|C|})}{2d^{|C|}}.$$
\qed
\end{proof}
\begin{remark}
We note that we did not try to optimize all constants in the statement of Lemma~\ref{lem:uniform-rounding}.
By choosing all parameters in our proof appropriately, it is possible to show that for every constant $\varepsilon > 0$, there is
a randomized rounding scheme, $\delta > 0$ and $d_0$ such that for every instance of \CSPd\ with $d \geq d_0$ the probability 
that each clause $C$ is satisfied is at least $\min((1-\varepsilon)\|z_C\|^2 \cdot |C|\, d, \delta \cdot e^{\delta |C|})/ d^{|C|}$.
\end{remark}
\section{Rounding Arbitrary SDP Solutions}
In this section, we show how to round an arbitrary SDP solution.
\begin{lemma}\label{lem:non-uniform-rounding}
There is a randomized polynomial-time algorithm that given an instance $\calI$ of the \CSPd\ problem (with $d\geq 113$)
and an SDP solution,
outputs an assignment $x_u$ such that for every clause $C \in \calC$:
$$\Prob{C \text{ is satisfied by } x_u} \geq 
\frac{\min(\|z_C\|^2 |C| d/64, 2e^{|C|/8})}{4d^{|C|}}.$$
\end{lemma}
\begin{proof}
For every index $u$, we sort all vectors $u_i$ according to their length. Let $S_u$ be the indices of $\lceil d/2\rceil$ 
shortest vectors among $u_i$, and $L_u = [d]\setminus S_u$  be the indices of $\lfloor d/2\rfloor$ 
longest vectors among $u_i$ (we break ties arbitrarily).
For every clause $C$ let $r(C) = |\set{(u,i)\in C: i\in S_u}|$.

\begin{claim}\label{claim:get_uniform}
 For every $i\in S_u$, we have $\|u_i\|^2 \leq 1/|S_u|$.
\end{claim}
\begin{proof}
Let $i\in S_u$. Note that $\|u_i\|^2 + \sum_{j\in L_u} \|u_j\|^2 \leq 1$ (this follows from SDP constraints). There are at least $\lceil d/2\rceil$
terms in the sum, and  $\|u_i\|^2$ is the smallest among them (since $i\in S_u$). Thus
$\|u_i\|^2 \leq 1/\lceil d/2\rceil = 1/|S_u|$.
\qed
\end{proof}

We use a combination of two rounding schemes: one of them works well on clauses $C$ with $r(C) \geq |C|/4$,
the other on clauses $C$ with $r(C) \leq |C|/4$.

\begin{lemma}\label{lem:reduction}
There is a polynomial-time randomized rounding algorithm 
that given an \CSPd\ instance $\calI$ with $d\geq 113$ outputs an assignment $x_u$ such that every 
clause $C$ with $r(C) \geq |C|/4$ is satisfied with probability at least
$$\frac{\min(\|z_C\|^2\, |C|\, d/64, e^{|C|/4})}{2d^{|C|}}.$$
\end{lemma}
\begin{proof}
We will construct a sub-instance $\calI'$ with a uniform SDP solution and then solve $\calI'$ using Lemma~\ref{lem:uniform-rounding}.
To this end, we first construct a partial assignment $x_u$.
For every $u\in X$, with probability $|L_u|/d = \lfloor d/2\rfloor /d$, we assign a value to $x_u$ uniformly 
at random from $L_u$; with probability $1 - |L_u|/d = |S_u| /d$, we do not assign any value to $x_u$.
Let $A =\set{u: x_u \text{ is assigned}}$. Let us say that a clause $C$ \textit{survives} 
the partial assignment step if for every $(u,i) \in C$ either $u\in A$ and $i = x_u$, or $u\notin A$ and $i \in S_u$.

The probability that a clause $C$ survives is 
\begin{align*}
\prod_{(u,i) \in C, i\in L_u}  & \Prob{x_u \text{ is assigned value } i}  \prod_{(u,i) \in C, i\in S_u}
\Prob{x_u \text{ is unassigned}} = \\
& \left(\frac{\lfloor d/2\rfloor}{d} \cdot \frac{1}{\lfloor d/2\rfloor}\right)^{|C| - r(C)}
\cdot 
\left(\frac{\lceil d/2\rceil}{d}\right)^{r(C)} = 
\frac{\lceil d/2\rceil^{r(C)}}{d^{|C|}}
.
\end{align*}

For every survived clause $C$, let $C' = \set{(u,i): u\notin A}$.
Note that for every $(u,i)\in C'$, we have $i\in S_u$.
We get a sub-instance $\calI'$ of our problem on the set of unassigned variables $\set{x_u: u\notin A}$
with the set of clauses $\set{C': C\in\calC \text{ survives}}$. The length of each clause $C'$
equals $r(C)$. In sub-instance $\calI'$, we require that each variable $x_u$ takes values in $S_u$.
Thus $\calI'$ is an instance of MAX CSP${}_{d'}$ problem with $d' = |S_u| = \lceil d/2\rceil$.

Now we transform the SDP solution for $\calI$ to an SDP solution for $\calI'$: 
we let $z_{C'} = z_C$ for survived clauses $C$,  remove vectors $u_i$ for all $u\in A$, $i\in[d]$
and remove vectors $z_C$ for non-survived clauses $C$.
By Claim~\ref{claim:get_uniform}, this SDP solution is a uniform solution for $\calI'$ 
(i.e. $\|u_i\| \leq 1/d'$ for every $u\notin A$ and $i\in S_i$; note that $\calI'$ has 
alphabet size $d'$). We run the rounding algorithm from 
Lemma~\ref{lem:uniform-rounding}. The algorithm assigns values to unassigned variables $x_u$.
For every \textit{survived} clause $C$, we get
\ifconf
\else
\begin{align*}
\Prob{C \text{ is satisfied by } x_u} &= \Prob{C' \text{ is satisfied by } x_u} \geq \frac{\min(\|z_C\|^2 |C'| d'/8, e^{|C'|})}{2{d'}^{|C'|}}\\
&= \frac{\min(\|z_C\|^2 r(C) d'/8, e^{r(C)})}{2{d'}^{r(C)}}\geq \frac{\min(\|z_C\|^2 |C| d/64, e^{|C|/4})}{2{d'}^{r(C)}}.
\end{align*}
\fi
Therefore, for every clause $C$,
\begin{align*}
\Prob{C \text{ is satisfied by } x_u} & \geq 
\Prob{C \text{ is satisfied by } x_u \;|\; C\text{ survives}} \Prob{C \text{ survives}} \\
& \geq 
\frac{\min(\|z_C\|^2 |C| d/64, e^{|C|/4})}{2{d'}^{r(C)}}
\times
\frac{\lceil d/2\rceil^{r(C)}}{d^{|C|}}
\\
&= \frac{\min(\|z_C\|^2 |C| d/64, e^{|C|/4})}{2d^{|C|}}.
\end{align*}
\qed
\end{proof}

Finally, we describe an algorithm for clauses $C$ with $r(C) \leq |C|/4$.
\begin{lemma}\label{lem:if-r-is-small}
There is a polynomial-time randomized rounding algorithm 
that given an \CSPd\ instance $\calI$ 
outputs an assignment $x_u$ such that every 
clause $C$ with $r(C) \leq |C|/4$ is satisfied with probability at least
$e^{|C|/8}/d^{|C|}$.
\end{lemma}
\begin{proof}
We do the following independently for every vertex $u\in X$. With probability
$3/4$, we choose $x_u$ uniformly at random from $L_u$; with probability 
$1/4$, we choose $x_u$ uniformly at random from $S_u$. The probability that a clause $C$ with $r(C) \leq |C|/4$ 
is satisfied equals 
\ifconf
\else
\begin{align*}
\prod_{(u, i) \in C, i\in L_u} \frac{3}{4|L_u|} & \prod_{(u, i) \in C, i\in S_u} \frac{1}{4|S_u|} 
= \frac{1}{d^{|C|}} \cdot \left(\frac{3d}{4|L_u|}\right)^{|C|-r(C)} \left(\frac{d}{4|S_u|}\right)^{r(C)}\\
&\geq 
\frac{1}{d^{|C|}} \cdot \left(\frac{3d}{4|L_u|}\right)^{3|C|/4} \left(\frac{d}{4|S_u|}\right)^{|C|/4} 
\geq
\frac{1}{d^{|C|}} \cdot \left(\left(\frac{3}{2}\right)^{3/4} \left(\frac{d}{2(d+1)}\right)^{1/4}\right)^{|C|}.
\end{align*}
\fi
Note that 
$\left(\frac{3}{2}\right)^{3/4} \left(\frac{d}{2(d+1)}\right)^{1/4} \geq 
\left(\frac{3}{2}\right)^{3/4} \left(\frac{113}{2\cdot 114}\right)^{1/4} \geq e^{1/8}$.
Therefore, the probability that the clause is satisfied is at least $e^{|C|/8}/d^{|C|}$.
\qed
\end{proof}

We run the algorithm from Lemma~\ref{lem:reduction} with probability $1/2$ and the algorithm from Lemma~\ref{lem:if-r-is-small}
with probability $1/2$. Consider a clause $C\in \calC$. If $r(C) \geq |C|/4$, we satisfy $C$ with probability at least
$\frac{\min(\|z_C\|^2 |C| d/64, e^{|C|/4})}{{4d^{|C|}}}$. If $r(C) \leq |C|/4$, we satisfy $C$ with probability at least
$e^{|C|/8}/(2d^{|C|})$. So we satisfy every clause $C$ with probability at least 
$\frac{\min(\|z_C\|^2 |C| d/64, \, 2e^{|C|/8})}{4d^{|C|}}$.
\qed
\end{proof}
\section{Approximation Algorithm for \texorpdfstring{\kCSPd}{k-CSP}}
In this section, we present the main result of the paper.

\begin{theorem}\label{thm:main}
There is a polynomial-time randomized approximation algorithm for \kCSPd\ that given 
an instance $\calI$ finds an assignment that satisfies at least $\Omega(\min(kd, e^{k/8})\, OPT(\calI)/d^k)$
clauses with constant probability. 
\end{theorem}
\begin{proof}
If $d\leq 113$, we run the algorithm of Charikar, Makarychev and Maka\-rychev~\cite{CMM06} and get $\Omega(k/d^k)$
approximation. So we assume below that $d\geq 113$. We also assume that $kd/d^k \geq 1/|\calC|$,
as otherwise we just choose one clause from $\calC$ and find an assignment that satisfies it.
Thus $d^k$ is polynomial in the size of the input. 

We solve the SDP relaxation for the problem and run the rounding scheme from Lemma~\ref{lem:non-uniform-rounding}
$d^k$ times. We output the best of the obtained solutions.
By Lemma~\ref{lem:non-uniform-rounding}, each time we run the rounding scheme we get a solution with expected
value at least
\ifconf
\else
\begin{align*}
\sum_{C\in\calC} \frac{\min(\|z_C\|^2 |C| d/64, 2e^{|C|/8})}{4d^{|C|}} &\geq
\sum_{C\in\calC} \frac{\min(k d/64, 2e^{k/8})}{4d^{k}} \|z_C\|^2\geq
\frac{\min(kd/64, 2e^{k/8})}{4d^k} SDP(\calI)\\
&\geq  \frac{\min(kd/64, 2e^{k/8})}{4d^k} OPT(\calI).
\end{align*}
\fi
Denote $\alpha = \frac{\min(kd/64, 2e^{k/8})}{4d^k}$. Let $Z$ be the random variable equal to the number of
satisfied clauses. Then $\E{Z} \geq \alpha OPT(\calI)$, and $Z \leq OPT(\calI)$ (always). Let $p = \Prob{Z\leq \alpha OPT(\calI)/2}$.
Then\
$$ p \cdot (\alpha OPT(\calI)/2) + (1 - p) \cdot OPT(\calI)  \geq \E{Z}\geq \alpha OPT(\calI).$$
So $p \leq \frac{1 - \alpha}{1-\alpha/2} = 1 - \frac{\alpha}{2-\alpha}$.
So with probability at least $1-p \geq \frac{\alpha}{2-\alpha}$, we find a solution of value at least $\alpha OPT(\calI)/2$ in one iteration.
Since we perform $d^k > 1/\alpha$ iterations, we find a solution of value at least $\alpha OPT(\calI)/2$ with constant probability.
\qed
\end{proof}

\appendix
\section{Proof of Lemma~\ref{lem:gaussian-scaling}}
In this section, we prove Lemma~\ref{lem:gaussian-scaling}.
We will use the following fact.
\begin{lemma}[see e.g.~\cite{CMMug}]\label{lem:NormalDistrib} 
For every $t > 0$,
$$  \frac{2t}{\sqrt{2\pi}\,(t^2+1)} e^{-\frac{t^2}{2}} < \bar{\Phi}(t) < \frac{2}{\sqrt{2\pi}\,t}
e^{-\frac{t^2}{2}}.$$
\end{lemma}

\noindent\textbf{Lemma~\ref{lem:gaussian-scaling}}.
\textit{For every $t > 0$ and $\beta \in (0,1]$, we have}
$$\bar{\Phi}(\beta t) \leq \bar{\Phi}(t)^{\beta^2}.$$
\begin{proof}
Rewrite the inequality we need to prove as follows:
$(\bar{\Phi}(\beta t))^{1/\beta^2} \leq \bar{\Phi}(t)$.
Denote the left hand side by $f(\beta,t)$:
$$f(\beta,t) = \bar{\Phi}(\beta t)^{1/\beta^2}.$$
We show that for every $t > 0$, $f(\beta, t)$ is strictly increasing function
as a function of $\beta\in (0, 1]$. Then,
$$(\bar{\Phi}(\beta t))^{1/\beta^2} = f(\beta) < f(1) = \bar{\Phi}(t).$$
We first prove that $\frac{\partial f(1,t)}{\partial \beta} > 0$.
Write,
$$\frac{\partial f(1,t)}{\partial \beta} = -2 \log (\bar{\Phi}(t))\, \bar{\Phi}(t)
+ t \bar{\Phi}'(t) =
-2 \log (\bar{\Phi}(t))\, \bar{\Phi}(t)
- \frac{2t\, e^{-t^2/2}}{\sqrt{2\pi}}.$$

Consider three cases. If $t\geq \sqrt{\frac{2e}{\pi}}$, then,
by Lemma~\ref{lem:NormalDistrib},
$$\bar{\Phi}(t) < \frac{2}{\sqrt{2\pi} t} e^{-t^2/2}\leq e^{-1/2}e^{-t^2/2}
= e^{-(t^2+1)/2}.$$
Hence, $-2 \log (\bar{\Phi}(t)) > (t^2 + 1)$, and by Lemma~\ref{lem:NormalDistrib},
$$-2 \log (\bar{\Phi}(t))\, \bar{\Phi}(t) >
(t^2 + 1)\, \bar{\Phi}(t) >
 \frac{2t\, e^{-t^2/2}}{\sqrt{2\pi}}.$$
If $t < \sqrt{\frac{2e}{\pi}}$, then let
$\rho(x) = -\log x / (1 - x)$ for $x\in (0,1)$ and write,
$$-\log \bar{\Phi}(t) = \rho(\bar{\Phi}(t)) \cdot (1 - \bar{\Phi}(t)) = \frac{\rho(\bar{\Phi}(t))}{\sqrt{2\pi}}\int_{-t}^t e^{-x^2/2} dx \geq
\frac{2\rho(\bar{\Phi}(t)) te^{-t^2/2}}{\sqrt{2\pi}}.$$
Hence,
$$\frac{\partial f(1,t)}{\partial \beta} =
-2 \log (\bar{\Phi}(t))\, \bar{\Phi}(t)
- \frac{2t\, e^{-t^2/2}}{\sqrt{2\pi}}
\geq \frac{2te^{-t^2/2}}{\sqrt{2\pi}} \times (2\rho(\bar{\Phi}(t)) \bar{\Phi}(t) - 1).$$

For $x\in [1/3,1]$, $2\rho(x) x >1$, since the function $\rho(x) x$ is increasing and 
$\rho(1/3) >3/2$. Hence $2\rho(\bar{\Phi}(t)) \bar{\Phi}(t) > 1$, if $\bar{\Phi}(t) \geq 1/3$.

The remaining case is $t < \sqrt{\frac{2e}{\pi}}$ and $\bar{\Phi}(t) < 1/3$. Then,
$\bar{\Phi}(t) \geq \bar{\Phi}(\sqrt{\frac{2e}{\pi}}) > 1/6$ and hence
$\bar{\Phi}(t)\in (1/6,1/3)$. Since the function $-x\log x$ is increasing on the interval
$(0,e^{-1})$,
$$-2 \log (\bar{\Phi}(t))\, \bar{\Phi}(t)> -2 \log (1/6)\cdot \frac{1}{6} > \frac{1}{2}.$$
The function $te^{-t^2/2}$ attains its maximum at $t=1$, thus
$$\frac{2t\, e^{-t^2/2}}{\sqrt{2\pi}}\leq \frac{2e^{-1/2}}{\sqrt{2\pi}} < \frac{1}{2}.$$

We get
$$\frac{\partial f(1,t)}{\partial \beta} =
-2 \log (\bar{\Phi}(t))\, \bar{\Phi}(t)
- \frac{2t\, e^{-t^2/2}}{\sqrt{2\pi}}>0.$$

Since $\frac{\partial f(1,t)}{\partial \beta} > 0$, for every $t > 0$, there exists
$\varepsilon_0 > 0$, such that for all $\varepsilon \in (0, \varepsilon_0)$,
$f(1-\varepsilon,t) < f(1,t)$. Particularly, for $t' = \beta t$,
$$f(\beta, t) = f(1, t')^{1/\beta^2} \geq f(1-\varepsilon, t')^{1/\beta^2}
= f((1-\varepsilon)\beta, t).$$
\qed
\end{proof}

\ifconf\else
\section{Improved Approximation Factor for Boolean Max \texorpdfstring{$k$}{k}-CSP} ~\label{sec:boolean-CSP}
In this section, we present an approximation algorithm for the boolean Maximum $k$-CSP problem, MAX $k$-CSP${}_2$.
The algorithm has approximation factor $0.626612\, k/2^k$ if $k$ is sufficiently large. This bound improves the previously best known bound 
of $0.44\, k/2^k$~\cite{CMM06} (if $k$ is sufficiently large).

Our algorithm is a slight modification of the algorithm for rounding uniform solutions of MAX $k$-CSP${}_d$.
We use the SDP relaxation presented in Section~\ref{sec:prelim}. 
Without loss of generality, we will assume below that all clauses have length exactly $k$. 
If a clause $C$ is shorter, we can introduce $k - |C|$ new variables and append them to $C$. 
This transformation will not change the value of the instance.

First, we describe a rounding scheme for an SDP solution $\set{u_1, u_2}_{u\in X} \cup  \set{z_C}_{C\in\calC}$. 
\begin{lemma} \label{lem:boolean-CSP-rounding}
There is a polynomial-time randomized rounding algorithm such that
for every clause $C\in \calC$ the probability that the  
algorithm satisfies $C$
is at least 
$$\frac{1}{2^k \sqrt{2\pi k}} \int_0^\infty h_{\beta}(t)^k dt, \text{ where } h_{\beta}(t)  = 2\Phi(\beta t)\, e^{-t^2/2},$$
and $\beta = \sqrt{k}\, \|z_C\|_2$. 
\end{lemma}
\begin{proof}
We round the SDP solution as follows.

\rule{0pt}{12pt}
\hrule height 0.8pt
\rule{0pt}{1pt}
\hrule height 0.4pt
\rule{0pt}{6pt}

\noindent \textbf{SDP Rounding Scheme for MAX $k$-CSP${}_2$}
\medskip

\noindent \textbf{Input:} an instance of MAX $k$-CSP${}_2$ and an SDP solution.

\noindent \textbf{Output:} an assignment $\{x_u\}$.
\begin{itemize}\compactify
\item Choose a random Gaussian vector $g$ so that every component of $g$ is distributed as a Gaussian variable with mean 0 and variance 1,
and all components are independent.
\item For every $u\in V$, let $x_u= \argmax_i \langle u_i, g\rangle$. 
\end{itemize}
\rule{0pt}{1pt}
\hrule height 0.4pt
\rule{0pt}{1pt}
\hrule height 0.8pt
\rule{0pt}{12pt}

Consider a clause $C\in \calC$. We assume without loss of generality that $C = \set{(u,1): u\in \supp(C)}$.
Let $\gamma_C = \langle z_C, g\rangle$ and $\gamma_u = \langle u_2 - u_1 + z_C, g\rangle$ for $u\in \supp(C)$.
Note that for $u\in \supp(C)$,
\begin{align*}
\Var[\gamma_C] &= \|z_C\|^2 = \beta^2/k,\\
\Var[\gamma_u] &= \|u_2 - u_1 + z_C\|^2 = \|u_1\| + \|u_2\|^2 + \|z_C\|^2 - 2\langle u_1, z_C\rangle = \|u_1\| + \|u_2\|^2 - \|z_C\|^2 \leq 1,\\
\E{\gamma_C \gamma_u} &= \langle z_C, u_2 - u_1 + z_C\rangle = \langle z_C, u_2 \rangle - \langle z_C, u_1 \rangle + \langle z_C, z_C \rangle = 
 0 - \|z_C\|^2 + \|z_C\|^2 = 0.
\end{align*}
Therefore, all random variables $\gamma_u$, for $u\in\supp C$, are independent from $\gamma_C$.
The probability that $C$ is satisfied equals
\begin{align*}
\Prob{C \text{ is satisfied}} &= \Prob{\langle u_1, g\rangle > \langle u_2, g\rangle \text{ for every } u\in \supp(C)}\\
&=\Prob{\gamma_C > \gamma_u      \text{ for every } u\in \supp(C)} \geq \Prob{|\gamma_u| < \gamma_C  \text{ for every } u\in \supp(C)}\\
&=\EE{\gamma_C}{\Prob{|\gamma_u| \leq \gamma_C \text{ for every } u\in \supp(C)\; | \; \gamma_C}} \\
& \stackrel{\text{let } t = \gamma_C/\beta}{=} 
\frac{1}{\sqrt{2\pi k}}\int_{t=0}^\infty \Prob{|\gamma_u| \leq \beta t \text{ for every } u\in \supp(C)} e^{-k t^2/2} dt.
\end{align*} 
We use here that $\Var[\gamma_C/\beta] = 1/k$. 
By \v{S}id\'{a}k's Theorem (Theorem~\ref{thm:sidak}), we have
\begin{align*}
\Prob{|\gamma_u| \leq \beta t \text{ for every } u\in \supp(C)} &\geq \prod_{u\in \supp(C)}
\Prob{|\gamma_u| \leq \beta t} 
= \prod_{u\in \supp(C)}  \Phi(\beta t/\sqrt{\Var[\gamma_u]}) \\
{}& \geq \prod_{u\in \supp(C)} \Phi(\beta t) = \Phi(\beta t)^k.
\end{align*}
We conclude that 
$$\Prob{C \text{ is satisfied}} \geq \frac{1}{2^k \sqrt{2\pi k}} \int_0^\infty h_{\beta}(t)^k dt.$$
\qed
\end{proof}

Let $g(\beta) = \max_{t\in\mathbb{R}} h_\beta(t)$ ($h_\beta$(t) attains its maximum since $h_\beta(t)\to 0$ as $t\to\infty$).
Note that $g(\beta)$ is an increasing function since $h_\beta(t)$ is an increasing function of $\beta$ for every fixed $t$. 
Additionally, $g(0) = 0$ and $\lim_{\beta \to \infty} g(\beta) = 2$ since $f(\beta, 1/\sqrt{\beta}) = 2\Phi(\sqrt\beta)e^{-1/(2\beta)} \to 2$
as $\beta \to \infty$, and for every $\beta$ and $t$, $f(\beta, t) \leq 2$. Therefore, $\beta^{-1}$ is defined on $[0, 2)$.
Let $\beta_0 = g^{-1}(1)$. It is easy to check numerically that $\beta_0 \in (1.263282, 1.263283)$. 
\begin{figure}
\begin{center}
\includegraphics[scale=0.75]{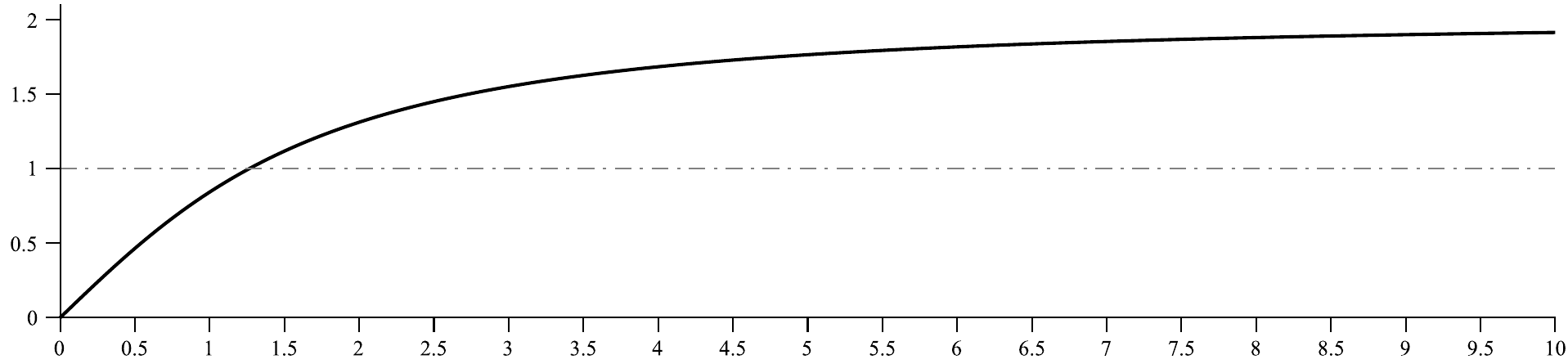}
\end{center}
\caption{The figure shows the graph of $g(t)$. We note that $g(t) > 1$ when $t > \beta_0 \approx 1.263282$. }
\end{figure}

\begin{claim} \label{claim:boolean-success}
For every $\beta > \beta_0$ there exists $k_0$ (which depends only on $\beta$)
such that if $k \geq k_0$ and $\|z_C\| \geq \beta/\sqrt{k}$ then the probability that the algorithm 
from Lemma~\ref{lem:boolean-CSP-rounding} returns an assignment that
satisfies $C$ is at least $k^2/2^{k}$. 
\end{claim}
\begin{proof}
Let $\varepsilon_1 = (g(\beta) - 1)/2 > 0$. Let $\varepsilon_2$ be the measure of the set
$\set{t:h_\beta(t) > 1 + \varepsilon_1}$. Since $h_\beta(t)$ is continuous, $\varepsilon_2 > 0$.

The probability that $C$ is satisfied is at least
$$\frac{1}{2^k \sqrt{2\pi k}} \int_0^\infty h_{\beta}(t)^k dt\geq \frac{\varepsilon_2 (1+\varepsilon_1)^k}{2^k \sqrt{2\pi k}}.$$
We choose $k_0$ so that for every $k \geq k_0$
$$\varepsilon_2 (1+\varepsilon_1)^{k} \geq \sqrt{2\pi k} \cdot k^2.$$
Then if $k \geq k_0$ the probability that the clause is satisfied is at least $k^2/2^{k}$.
\qed
\end{proof}

Now we are ready to describe our algorithm. 
\begin{theorem} There is a randomized approximation algorithm for the boolean MAX  $k$-CSP problem with approximation guarantee 
$\alpha_k k/2^k$ where $\alpha_k \to \alpha_0 \geq 0.626612$ as $k\to \infty$ and $\alpha_0 = 1/\beta_0^2$. 
(Here, as above, $\beta_0$ is the solution of the equation $g(\beta) = 1$ where $g(\beta) = \max_{t\in \mathbb R} 2\Phi(\beta t) e^{-t^2/2}$.)
\end{theorem}
\begin{proof}
The algorithm with probability $p = 1/k$ rounds the SDP solution
as described in Lemma~\ref{lem:boolean-CSP-rounding}, with probability $1-p$, it choses a completely random solution. 

Let $\alpha < \alpha_0$. We will show that if $k$ is large enough, every clause is satisfied with probability at least
$\alpha k/2^k$. 
Let $\beta = (\beta_0 + \alpha^{-1/2}) / 2 \in (\beta_0, \alpha^{-1/2})$. Let $k_0$ be as in Claim~\ref{claim:boolean-success}.
Suppose that $k \geq \max(k_0, (1 - \alpha\beta^2)^{-1})$. 

Consider a clause $C$. 
We show that the algorithm satisfies $C$ with probability at least 
$\frac{\alpha\|z_C\|^2 k}{2^k}$. Indeed, we have:
\begin{itemize}
\item If $\|z_C\| < \beta/\sqrt{k}$, the clause is satisfied  with probability at least
$(1-p)/2^k \geq \frac{(1-p)k \,\|z_C\|^2}{\beta^2 2^k} \geq \alpha\|z_C\|^2 k/2^k$.
\item If $\|z_C\| \geq \beta/\sqrt{k}$, the clause is satisfied  with probability at least 
$p\cdot k^2 /2^k = k/2^k \geq k\,\|z_C\|^2/2^k$.
\end{itemize}

We conclude that the algorithm finds a solution that satisfies at least 
$$\frac{\alpha k}{2^k} \sum_{C\in \calC} \,\|z_C\|^2 = \frac{\alpha k}{2^k} \cdot SDP \geq \frac{\alpha k}{2^k} \cdot OPT$$
clauses in expectation. By running this algorithm polynomially many times (as we do in Theorem~\ref{thm:main}) we can find a solution
of value at least $\alpha' k\, OPT/2^k$ for every constant $\alpha' < \alpha$ w.h.p.
\qed
\end{proof}
\fi
\end{document}